%% file: main.tex
\documentclass[10pt]{llncs}
\pdfoutput=1

\usepackage{xspace}
\usepackage[english]{babel}
\usepackage[latin1]{inputenc}

\usepackage{subfigure}
\usepackage{amsmath,amssymb}
\usepackage{enumerate}

\usepackage{times}

\usepackage{floatflt}

\usepackage{graphicx}

\usepackage{tikz}
\usetikzlibrary{automata,positioning,fit,calc,backgrounds}
\tikzset{initial text={},
    every state/.style={circle,minimum size=.6cm,draw=blue!50,very thick,fill=blue!20},
    acc/.style={double distance=1.5pt},
    secret/.style={minimum size=.4cm,draw=red!50,very thick,fill=red!20,rectangle},
    node distance=1.5cm,on grid,auto,
    bend angle=65}

\def\ie{{i.e.},~}
\def\eg{{e.g.},~}
\def\st{{s.t.}~}
\def\emptyset{\varnothing}

\def\Tr{\textit{Tr}}
\def\tr{\textit{tr}}

\def\prefaulty{\textit{PreFaulty}}

\def\nonfaulty{\textit{NonFaulty}}

\def\runs{\textit{Runs}} 
\def\lang{{\cal L}} 

\def\projs{\boldsymbol{\pi}}
\newcommand{\proj}[1]{\projs_{/#1}} 

\newcommand{\vect}[1]{\mathbf{#1}}

\newtheorem{prob}{Problem}  

\newcommand{\setN}{\mathbb N}
\newcommand{\setR}{\mathbb R}
\newcommand{\setB}{\mathbb B}
\newcommand{\setZ}{\mathbb Z}
\newcommand{\setQ}{\mathbb Q}

\def\calA{{\cal A}}
\def\calB{{\cal B}}

\def\calC{{\cal C}}

\def\cc{\calC}

\def\enabled{\textit{en}}

\def\endef{\ifmmode\squareforged\else{\unskip\nobreak\hfil
\penalty50\hskip1em\null\nobreak\hfil$\blacksquare$
\parfillskip=0pt\finalhyphendemerits=0\endgraf}\fi}

\def\tauac{\tau}

\newcommand{\dur}{{\textit{Dur}}} 
\def\inv{\textit{Inv}}

\def\tw{\textit{TW\/}}

\def\untimed{\textit{Unt}}

\def\true{\mbox{\textsc{true}}}
\def\false{\mbox{\textsc{false}}}

\def\rg{\textit{RG}}

\def\tgt{\textit{tgt}}

\def\uppaal{\textsc{Uppaal}\xspace}

\newcommand{\stack}[1]{\begin{tabular}{c}#1\end{tabular}}
\newcommand{\prefix}[1]{\overline{#1}}

\title{\LARGE \bf Predictability of Event Occurrences \\ in Timed Systems}

\author{Franck Cassez\inst{1} \and Alban Grastien\inst{2}}

\institute{
NICTA\thanks{NICTA is funded by the Australian Government
    as represented by the Department of Broadband, Communications and
    the Digital Economy and the Australian Research Council through
    the ICT Centre of Excellence program.} and UNSW, Sydney
\and NICTA and ANU, Canberra \\ Australia}

\begin{document}
\pagestyle{plain}
\maketitle
  
\thispagestyle{empty}

\begin{abstract} 
  We address the problem of predicting events' occurrences in
  partially observable timed systems modelled by timed automata.  Our
  contribution is many-fold: 1) we give a definition of bounded
  predictability, namely $k$-predictability, that takes into account
  the minimum delay between the prediction and the actual event's
  occurrence; 2) we show that $0$-predictability is equivalent to the
  original notion of predictability of S.~Genc and S.~Lafortune; 3) we
  provide a necessary and sufficient condition for $k$-predictability
  (which is very similar to $k$-diagnosability) and give a simple
  algorithm to check $k$-predictability; 4) we address the problem of
  predictability of events' occurrences in timed automata and show
  that the problem is PSPACE-complete.
\end{abstract}
\medskip

\input{intro}

\input{prelim}

\input{predic-problems-nsc}

\input{predic-fa}

\input{predic-ta}

\input{conclusion}

\bibliography{diagnosis}

\newpage

\appendix

\section{Proof of Theorem~\ref{thm-equiv}}

\begin{proof}
{\it if Part.}
Assume there exists a $0$-predictor $P$ for $A$ and
Equation~\eqref{eq-pred-lafortune} does not hold.  Then $\forall n,
\exists w \in L_f, \forall t \in \overline{w}$, $\mathbf{P}(t)$ does
not hold.
Let $t = w^{-0}=w$.  As $\mathbf{P}(t)$ does not hold: $\exists u \in
L_{\neg f}, \exists v \in \lang^*(A)/u, \projs(u)=\projs(t) \wedge |v|
\geq n$ but $|v|_f = 0$.  Assume we have $n \geq |Q|$, the number of
states of $A$.  Then $v$ has a cycle (pumping Lemma) and can be
written $v = x. y.  z$ with $|x. y^j. z|_f= 0, x .y^j. z \in
\lang^*(A)/u, \forall j \geq 0$ and thus we can build $v' =
x. y^\omega \in \lang(A)^\omega/u$ s.t. $|v'|_f=0$.  It follows that
$u \ v' \in \lang^\omega_{\neg f}$.
  We have: 1) $w \in
L^{-0}_f$, 2) $w' \in L^\omega_{\neg f}$.  Moreover
$\projs(w^{-0})=\projs(t)$ and $P(\projs(w^{-0}))=1$ because $P$ is a
$0$-predictor.  But $\projs(t) = \projs(u)$ and $ u \in
\prefix{L^\omega_{\neg p}}$ which entails $P(\projs(u))=0$ which is a
contradiction.

\noindent{\it Only if.}
Assume Equation~\eqref{eq-pred-lafortune} holds.
Define the mapping $P$ as follows: 
\begin{itemize}
\item $\forall w \in L^{-0}_f, P(\projs(w))=1$ and
\item $\forall w \in \prefix{L^\omega_{\neg p}}$, $P(\projs(w)) = 0$.
\end{itemize}
We can show that $P$ is a $0$-predictor \ie it is well-defined.  On
the contrary assume there exists $r = \projs(w_1)$ with $w_1 \in
L^{-0}_f$ and $r=\projs(w_2)$ with $w_2 \in \prefix{L^\omega_{\neg p}}$.
We can  show that Equation~\eqref{eq-pred-lafortune} cannot hold which
is a contradiction.
Take $n \in \setN$. $w_1 \in L^{-0}_f = L_f$ and, by
Equation~\eqref{eq-pred-lafortune}, there must exist $t \in
\prefix{w_1}$ s.t.  $\mathbf{P}(t)$ holds.
But we can exhibit two words $u \in L_{\neg p}$ and $v \in
\lang^*(A)/u$ that falsify $\mathbf{P}(t)$. $t \in \prefix{w_1}$.  As
$\projs(w_1)= \projs(w_2)$, there exists $w'_2 \in \prefix{w_2}$
s.t. $\projs(t) = \projs(w'_2)$.  $w'_2 \in \prefix{L^\omega_f}$
because $w_2 \in \prefix{L^\omega_f}$.
Take $u=w'_2$ and $v \in L_{\neg f}/w'_2$ with $|v| \geq n$ (exists as
$w'_2 \in \prefix{L^\omega_f}$).  We have $\projs(u)=\projs(t)$, $v
\in \lang^*(A)/u$, $|v| \geq n$ but $|v|_f = 0$ which contradicts
Equation~\eqref{eq-pred-lafortune}.  \qed
\end{proof}

\section{Proof  of  Lemma~\ref{lemma-inf}}
\begin{proof}
%
\noindent{\it If.}
Assume $\projs(L^{-\Delta}_f) \cap \projs(\prefix{L^\omega_{\neg f}})
\neq \varnothing$.  Let $w \in \projs(L^{-\Delta}_f) \cap
\projs(\prefix{L^\omega_{\neg f}})$.  Then $w = \projs(w_1) =
\projs(w_2)$ with $w_1 \in L^{-\Delta}_f$ and $w_2
\in\prefix{L^\omega_{\neg f}}$.  Moreover there exists some $w'_2 \in
L^{\omega}_{\neg f}$ such that $w_2.w'_2 \in L^{\omega}_{\neg f}$.  It
follows that $\projs(w_2.w'_2)$ is an infinite timed word because by
assumption every infinite timed has an infinite number of events in
$\Sigma_o$.  By definition of $L^{\omega,-\Delta}_f$, $w_1.w'_2 \in
L^{\omega,-\Delta}_f$. Moreover\footnote{The condition
  $\dur(w_1)=\dur(w_2)$ is only needed for TA. For FA, it does not
  hold but is not necessary to concatenate the words.}
$\dur(w_1)=\dur(w_2)$ and $ \projs(w_1) = \projs(w_2)$ and thus
$\projs(w_1.w'_2) = \projs(w_1).\projs(w'_2) = \projs(w_2.w'_2)$. This
entails $\projs(L^{\omega,-\Delta}_f) \cap \projs(L^\omega_{\neg f})
\neq \varnothing$.
 
\noindent{\it Only If.}
Now assume $w \in \projs(L^{\omega,-\Delta}_f) \cap
\projs(L^\omega_{\neg f}) \neq \varnothing$. We have $w =
\projs(w_1.w'_1) = \projs(w_2)$ with $w_1 \in L^{-\Delta}_f$, $w_2 \in
L^{\omega}_{\neg f}$.  Let $w'_2$ be a prefix of $w_2$ such that
$\projs(w'_2)=\projs(w_1)$ (such a prefix exists because
$\projs(w_1.w'_1) = \projs(w_2)$.)  Then $w'_2 \in
\prefix{L^{\omega}_{\neg f}}$ and $w_1 \in L^{-\Delta}_{f}$ and
$\projs(w'_2)=\projs(w_1)$ which entails that $ \projs(L^{-\Delta}_f)
\cap \projs(\prefix{L^\omega_{\neg f}}) \neq \varnothing$.  \qed
\end{proof}

\section{Proof  of  Lemma~\ref{lemma-equiv}}
\begin{proof}
\noindent{\it \fbox{$\supseteq$}}
Let $w \in \lang^\omega(A_1(\Delta) \times A_2) =
\lang^\omega(A_1(\Delta)) \cap \lang^\omega(A_2)$.  Then $w =
\projs(\tr(\rho_1))$ with $\rho_1 \in \runs^\omega(A_1(\Delta))$ and
$w = \projs(\tr(\rho_2))$ with $\rho_2 \in \runs^\omega(A_2)$.
We can write $\tr(\rho_1) = w_1.w'_1.w''_1$ with $w_1 \in
L^{-\Delta}_f$, $\dur(w'_1) \leq \Delta$ and $w''_1 \in (\setR_{\geq
  0} \times \Sigma_o)^\omega$ by construction of $A_1$ and its
accepting condition. It follows that $w \in
\projs(L^{\omega,-\Delta}_f) \cap \projs(L^\omega_{\neg f})$.

\noindent{\it \fbox{$\subseteq$}}
Let $w \in \projs(L^{\omega,-\Delta}_f) \cap \projs(L^\omega_{\neg
  f})$.  
Fig.~\ref{fig-split} depicts the following proof.  We can write $w =
\projs(w_1.w^+_1)$ with $w_1 \in L^{-\Delta}_f$, $w^+_1 \in
(\setR_{\geq 0} \times \Sigma_o)^\omega$.  We also have $w =
\projs(w_2.w^+_2)$ for some $w_2.w^+_2 \in L^\omega_{\neg f}$ and such
that $\projs(w_1) = \projs(w_2)$ and $\projs(w^+_1) = \projs(w^+_2)$.
Note also that $\projs(w_2.w^+_2) \in \lang^\omega(A_2)$ because we
assume every infinite timed word has an infinite number of $\Sigma_o$
actions.
\begin{figure}
\center
\begin{tikzpicture}
      \tikzset{node distance=1cm and 2.3cm, state/.style={minimum
    width=0.1cm,minimum height=0.1cm,red,arrow head=6mm},
  loope/.style={in=65,out=115,distance=1.2cm,loop},xscale=3} 

      \draw[->,thick] (0,2) --  (3,2);
      \draw (0,2) node[label=left:{$A_1$}]  {$\bullet$};
      \draw (1,2) node  {$\bullet$};

      \draw(0.5,2.3) node {$w_1$};
      \draw(1.5,2.3) node {$w^+_1$};

      \draw[->,thick] (1,1) --  (3,1);
      \draw (1,1) node  {$\bullet$};
      \draw (2,1) node[label={above,yshift=0.09cm,xshift=.5cm}:{END}]  {$\bullet$};
      \draw[->,dashed] (2.1,1.3) -- (2.024,1.045);
      \draw(1.5,1.0+0.3) node {$w'_2$};
      \draw(2.5,1.0+0.3) node {$w''_2$};

      \draw[->,thick] (0,0) --  (3,0);
      \draw (0,0) node[label=left:{$A_2$}]  {$\bullet$};
      \draw (1,0) node  {$\bullet$};
      \draw (2,0) node  {$\bullet$};

      \draw(0.5,-0.3) node {$w_2$};
      \draw(1.5,-0.3) node {$w'_2$};
      \draw(2.5,-0.3) node {$w''_2$};

      \draw[<->,dashed,thick](1,-1+.2) -- (2,-1+.2);
      \draw(1.5,-1.2) node {$\leq \Delta$ time units};

      \draw[-,dashed] (1,-1.2) -- (1,2.2);
      \draw[-,dashed] (2,-1.2) -- (2,2.2);

      \draw[->,thick](1,2) to[in=150,out=-150] node[left] {$y:=0$} (1-0.03,1);

\end{tikzpicture} 
\caption{Proof of Lemma~\ref{lemma-equiv}, $\subseteq$.}
\label{fig-split}
\end{figure}
\noindent As $w_1 \in L^{-\Delta}_f$, we can split $w^+_1$ into $w^+_1
= w'_1.w''_1$ with $\dur(w'_1) \leq \Delta$. We can split $w^+_2$
accordingly such that $w^+_2 = w'_2.w''_2$ and $\projs(w'_2) =
\projs(w'_1)$ and $\projs(w''_2)= \projs(w''_1)$ and $\dur(w'_2) =
\dur(w'_1) \leq \Delta$. Moreover $w_1.w'_2$ can be generated in
$A_1(\Delta)$ as follows: start with $w_1$, after $w_1$ switch to the
twin copy and reset $y$ and generate $w'_2$.  By playing $w'_1$ in the
original copy of $A$ in $A_1(\Delta)$ we reach a state $(l',v)$ where
$f$ is enabled: there is a transition $(l',g,f,R,l_f) \in E$ such that
$v \models g$.  By construction of $A_1(\Delta)$, playing $w'_2$ in
the twin copy in $A_1(\Delta)$ we reach an equivalent state
$(\tilde{l}',v)$ and a twin transition $(\tilde{l}',g \wedge y \leq
\Delta,\varepsilon,\{y\},\text{END})$.  As $\dur(w'_2) \leq \Delta$,
we must have $y \leq \Delta$ and this twin transition can be fired and
END is reachable. We can subsequently read $w''_2$ in $A_1(\Delta)$.
It follows that $\projs(w_1.w^+_2) \in \lang^\omega(A_1(\Delta))$ and
$w \in \lang^\omega(A_1(\Delta) \times A_2)$.  \qed
\end{proof}

\newcommand{\red}[1]{\textcolor{red}{#1}}
\newcommand{\green}[1]{\textcolor{green}{#1}}
\newcommand{\blue}[1]{\textcolor{blue}{#1}}
\newcommand{\white}[1]{\textcolor{white}{#1}}

\end{document}

%% file: intro.tex
\section{Introduction}

Monitoring and fault diagnosis aim at detecting defects that can occur
at run-time.  The monitored system is partially observable but a
formal model of the system is available which makes it possible to
build (offline) a monitor or a diagnoser.  Monitoring and fault
diagnosis for discrete event systems (DES) have been have been
extensively investigated in the last two
decades~\cite{Raja95,yoo-lafortune-tac-02,Jiang-01}.  Fault diagnosis
consists in detecting a fault \emph{as soon as possible} after it
occurred. It enables a system operator to stop the system in case
something went wrong, or reconfigure the system to drive it to a safe
state.  \emph{Predictability} is a strong version of diagnosability:
instead of detecting a fault after it occurred, the aim is to
\emph{predict} the fault before its occurrence.  This gives some time
to the operator to choose the best way to stop the system or to
reconfigure it.

In this paper, we address the problem of predicting event occurrences
in partially observable timed systems modelled by timed automata.
\smallskip

\noindent{\it The Predictability Problem.}
A timed automaton~\cite{AlurDill94} (TA) generates a timed language
which is a set of timed words which are sequences of pairs (event,
time-stamp).  Only a subset of the events generated by the system is
observable.  The objective is to predict occurrences of a particular
event (observable or not) based on the sequences of observable events.
Automaton $G$, Fig.~\ref{fig-example}, is a timed version of the
example of automaton $G_1$ of~\cite{automatica-genc-09}. The set of
observable events is $\{a,b,c\}$. We would like to predict event $f$
without observing event $d$.
\begin{figure}[hbtp]
  \centering
  \scalebox{0.9}{\begin{tikzpicture}[thick,node distance=1cm and 3cm]%
    \small
    \node[state,initial] (q_0) [label=-90:{$[x \leq 1]$}] {$l_0$}; 
    \node[state] (q_3) [above right=of q_0,label=-90:{$[x \leq 1]$}] {$l_3$};
    \node[state] (q_1) [below right=of q_0,label=-90:{$[x \leq 2]$}] {$l_1$};
    \node[state] (q_2) [right=of q_1,label=-90:{$[x \leq 3]$}] {$l_2$};
    \node[state] (q_4) [right=of q_3,label=-90:{$[x < 1]$}] {$l_4$};
    \node[state,secret] (q_f) [right=of q_2,label=-90:{$[x \leq 1]$}] {$l_f$};

    \path[->] (q_0) edge[swap,pos=0.66] node[yshift=.2cm] {\stack{$x=1$\\$a$\\$x:=0$}} (q_1)
    (q_0) edge[pos=0.65] node[yshift=-.2cm] {\stack{$x<1$\\$d$}} (q_3)
    (q_1) edge node {$x=2$,$c$,$x:=0$} (q_2)
    (q_3) edge node {$x\leq 1$,$a$,$x:=0$} (q_4)
    (q_4) edge[loop right] node {$b;x:=0$} (q_4)
    (q_2) edge node {$x \geq 2$,$f$,$x:=0$} (q_f)
    (q_f) edge[loop right] node {\stack{$a,b,c$\\$x:=0$}} (q_f); 
  \end{tikzpicture}
}
\caption{Example $G$ from~\cite{automatica-genc-09}.}
\label{fig-example}
\vspace*{-.6cm}
\end{figure}
First consider the untimed version of $G$ by ignoring the constraints
on clock $x$.  The untimed automaton can generate two types of events'
sequences: $d.a.b^*$ and $a.c.f.\{a,b,c\}^*$.  Because $d$ is
unobservable, after observing $a$ we do not know whether the system is
in location $l_4$ or $l_1$ and cannot predict $f$ as, according to our
knowledge, it is not bound to occur in all possible futures from
locations $l_4$ or $l_1$.  However, after the next observable event,
$b$ or $c$, we can make a decision: if we observe $a.c$, $G$ must be
in $l_2$ and thus $f$ is going to happen next. After observing $a.c$
we can predict event $f$.  Note that there is no quantitative duration
between occurrences of events in discrete event systems and thus we
can predict $f$ at a \emph{logical} time which is before $f$ occurs.
The time that separates the prediction of $f$ from the actual
occurrence of $f$ is measured in the number of discrete steps $G$ can
make. In this sense $G$ is $0$-predictable as when we predict $f$, it
is the next event to occur. The untimed version of $G$ is an
abstraction of a real system, and in the real system, it could be that
$f$ is going to occur 5 seconds after we observe $c$.

Timed automata enable us to capture quantitative aspects of real-time
systems.  We can use \emph{clocks} (like $x$) to specify constraints
between the occurrences of events.  Moreover \emph{invariants} (like
$[x \leq 1]$) ensure that $G$ changes location when the upper bound of
the invariant is reached. In the timed automaton $G$, the (infinite)
sequences with no $f$ are of the form
$(d,\delta_d)(a,\delta_a)(b,\delta_b)\cdots$ with $\delta_d < 1,
\delta_a \leq 1$ and $\delta_b < 2$. The sequences with event $f$ are
of the form $(a,1)(c,3)(f,\delta_f)$ with $5 \leq \delta_f \leq 6$.
Thus if we do \emph{not} observe a ``b'' within the first two time
units, we know that the system is in location $l_1$.  This implies
that $f$ is going to occur, and we know this at time $2$. But $f$ will
not occur before $1 + 2$ time units, the time for $c$ to occur (from
time $2$) and the minimum time for $f$ to occur after $c$.  $G$ is
thus $3$-predictable.  In the sequel we formally define the previous
notions and give efficient algorithms to solve the predictability
problem.


\smallskip

\noindent{\it Related Work.}
Predictability for discrete event systems was first proposed by
S.~Genc and S.~Lafortune in~\cite{ifac-2006}. Later
in~\cite{automatica-genc-09} they gave two algorithms to decide the
predictability problem, one of them is a polynomial decision
procedure.  T.~J\'eron, H.~Mar\-chand, S.~Genc and
S.~Lafortune~\cite{ifac2008} extended the previous results to
occurrences of \emph{patterns} (of events) rather than a single event.
L. {Brand\'an Briones} and A. {Madalinski} in~\cite{sccc-11} studied
\emph{bounded} predictability without relating it to the notion
defined by S.~Genc and S.~Lafortune.

Predictability is closely related to \emph{fault
  diagnosis}~\cite{Raja95,yoo-lafortune-tac-02,Jiang-01}.  The
objective of fault diagnosis is to detect the occurrence of a special
event, a fault, which is unobservable, as soon as possible after it
occurs.  Fault diagnosis for timed automata has first been studied by
S.~Tripakis in~\cite{tripakis-02} and he proved that the diagnosis
problem is PSPACE-complete.  P.~Bouyer,
F.~Chevalier and D.~D'Souza~\cite{Bouyerfossacs05} later studied the problem of
computing a diagnoser with \emph{fixed} resources (a deterministic TA)
and proved that this problem is 2EXPTIME-complete.
To the best of our knowledge the predictability problem for TA has not
been investigated yet.

\smallskip

\noindent{\it Our Contribution.}
We give a new characterization of bounded predictability and show it
is equivalent to the definition of S.~Genc and S.~Lafortune.  This new
characterization is simple and dual to the one for the diagnosis
problem; we can derive easily algorithms to decide predictability,
bounded predictability, and to compute the largest anticipation delay
to predict a fault.
We also study the bounded predictability problem for TA and prove it
is PSPACE-complete.  We investigate implementability issues, \ie how
to build a \emph{predictor}, and solve the \emph{sampling
  predictability} problem which ensures an implementable predictor
exists. We show how to compute bounded predictability 
with \uppaal~\cite{uppaal-sttt}.

\smallskip

\noindent{\it Organization of the Paper.}
The paper is organized as follows: the next section recalls some
definitions: timed words, timed automata. Section~\ref{sec-def-nsc}
states the predictability problems for TA and Finite Automata (FA) and
presents a necessary and sufficient condition for bounded
predictability.  Section~\ref{sec-predic-fa} compares our definition
of predictability with the original one (by S.~Genc and S.~Lafortune)
and provides an algorithm (for finite automata) to solve the bounded
predictability problem and compute the largest bound.
Section~\ref{sec-predic-ta} studies the bounded predictability problem
for TA and implementation issues related to the construction of a
predictor. An example is also solved with \uppaal.
Omitted proofs are given in Appendix.


%% file: prelim.tex
\section{Preliminaries}
$\setB=\{\true,\false\}$ is the set of boolean values, $\setN$ the set
of natural numbers, $\setZ$ the set of integers and $\setQ$ the set of
rational numbers.  $\setR$ is the set of real numbers and $\setR_{\geq
  0}$ is the set of non-negative reals.

\subsection{Clock Constraints}
Let $X$ be a finite set of variables called \emph{clocks}.  A
\emph{clock valuation} is a mapping $v : X \rightarrow \setR_{\geq
  0}$. We let $\setR_{\geq 0}^X$ be the set of clock valuations over
$X$. We let $\vect{0}_X$ be the \emph{zero} valuation where all the
clocks in $X$ are set to $0$ (we use $\vect{0}$ when $X$ is clear from
the context).  Given $\delta \in \setR$, $v + \delta$ denotes the
valuation defined by $(v + \delta)(x)=v(x) + \delta$. We let $\cc(X)$
be the set of \emph{convex constraints} on $X$ which is the set of
conjunctions of constraints of the form $x \bowtie c$ with $c
\in\setN$ and $\bowtie \in \{\leq,<,=,>,\geq\}$. Given a constraint $g
\in \cc(X)$ and a valuation $v$, we write $v \models g$ if $g$ is
satisfied by $v$.  Given $R \subseteq X$ and a valuation $v$, $v[R]$
is the valuation defined by $v[R](x)=v(x)$ if $x \not\in R$ and
$v[R](x)=0$ otherwise.

\subsection{Timed Words}
A finite (resp. infinite) \emph{timed word} over $\Sigma$ is a word in
$\setR_{\geq 0}.(\Sigma.\setR_{\geq 0})^*$ (resp. $(\setR_{\geq
  0}.\Sigma)^\omega$).  We write timed words as $0.4\ a\ 1.0\ b\ 2.7 \
c \cdots$ where the real values are the durations elapsed between two
events: thus $c$ occurs at global time $4.1$.  We let $\dur(w)$ be the
duration of a timed word $w$ which is defined to be the sum of the
durations (in $\setR_{\geq 0}$) which appear in $w$; if this sum is
infinite, the duration is $\infty$.  Note that the duration of an
infinite word can be finite, and such words which still contain an
infinite number of events, are called \emph{Zeno} words.  An infinite
timed word $w$ is \emph{time-divergent} if $\dur(w) = \infty$.  We let
$\untimed(w)$ be the \emph{untimed} version of $w$ obtained by erasing
all the durations in $w$, \eg $\untimed(0.4\ a\ 1.0\ b\ 2.7 \ c \ 0)=abc$.
Given $w$ a timed word and $a \in \Sigma$, $|w|_a$ is the number of
occurrences of $a$ in $w$ ($\infty$ if $a$ occurs infinitely often in
$w$.)

$\tw^*(\Sigma)$ is the set of finite timed words over $\Sigma$,
$\tw^\omega(\Sigma)$, the set of infinite timed words and
$\tw^\infty(\Sigma)=\tw^*(\Sigma) \cup \tw^\omega(\Sigma)$.  We use
$\Sigma^*$ and $\Sigma^\omega$ for the corresponding sets of untimed
words.  A \emph{timed language} is any subset of
$\tw^\infty(\Sigma)$. For $L \subseteq \tw^\infty(\Sigma) $, we let
$\untimed(L)=\{ \untimed(w) \ | \ w \in L \}$.

For $w \in \tw^*(\Sigma)$ and $w' \in \tw^\infty(\Sigma)$, $w.w'$ is
the concatenation of $w$ and $w'$.  A finite timed word $w$ is a
prefix of $w' \in \tw^\infty(\Sigma)$ if $w' = w.w''$ for some $w''
\in \tw^\infty(\Sigma)$. In the sequel we also the prefix operator and
$\prefix{L}$ is the set of finite words that are prefixes of words in
$L$.

Let $\Sigma_1 \subseteq \Sigma$. $\proj{\Sigma_1}$ is the projection
of timed words of $\tw^\infty(\Sigma)$ over timed words of
$\tw^\infty(\Sigma_1)$.  When projecting a timed word $w$ on a
sub-alphabet $\Sigma_1 \subseteq \Sigma$, the durations elap\-sed
bet\-ween two events are set accordingly: $\proj{\{a,c\}}(0.4 \ a\
1.0\ b\ 2.7 \ c )=0.4 \ a \ 3.7 \ c$ (projection erases some events
but preserves the time elapsed between the non-erased events).  It
follows that $\proj{\Sigma_1}(w) = \proj{\Sigma_1}(w')$ implies that
$\dur(w) = \dur(w')$.  For $L \subseteq \tw^\infty(\Sigma) $,
$\proj{\Sigma_1}(L)=\{ \proj{\Sigma_1}(w) \ | \ w \in L\}$.

\subsection{Timed Automata}
Timed automata (TA) are finite automata extended with real-valued
clocks to specify timing constraints between occurrences of events.
For a detailed presentation of the fundamental results for timed
automata, the reader is referred to the seminal paper of R.~Alur and
D.~Dill~\cite{AlurDill94}.  As usual we use the symbol $\varepsilon$
to denote the silent (invisible) action in an automaton.

\noindent\begin{definition}[Timed Automaton]\label{def-ta} 
  A \emph{Timed Automaton} $A$ is a tuple $(L,$ $l_0,$ $X,\Sigma \cup
  \{\varepsilon \}, E, \inv, F, R)$ where:
$L$ is a finite set of  \emph{locations}; 
$l_0$ is the \emph{initial location};
$X$ is a finite set of \emph{clocks};
$\Sigma$ is a finite set of \emph{events};
$E \subseteq L \times\calC(X) \times \Sigma \cup \{ \varepsilon \}
\times 2^X \times L$ is a finite set of \emph{transitions}; for
$(\ell,g,a,r,\ell') \in E$, $g$ is the \emph{guard}, $a$ the
\emph{event}, and $r$ the \emph{reset} set;
$\inv : L \rightarrow  \calC(X)$ associates with each location an
\emph{invariant}; as usual we require the invariants to be
conjunctions of constraints of the form $x \preceq c$ with $\preceq
\in \{<,\leq\}$.
  $F \subseteq L$ and $R \subseteq L$ are respectively the
  \emph{final} and \emph{repeated} sets of locations. \endef
\end{definition}
A \emph{state} of $A$ is a pair $(\ell,v) \in L \times \setR_{\geq
  0}^X$.
%
%
A \emph{run} $\varrho$ of $A$ from $(\ell_0,v_0)$ is a
(finite or infinite) sequence of alternating \emph{delay} and
\emph{discrete} moves:
\begin{eqnarray*}
  \varrho & = & (\ell_0,v_0) \xrightarrow{\delta_0} (\ell_0,v_0 + \delta_0)
  \xrightarrow{a_1} (\ell_1,v_1)  \cdots \xrightarrow{a_{n}} (\ell_n,v_n)
  \xrightarrow{\delta_n} (\ell_n,v_n+ \delta_n) \cdots 
\end{eqnarray*}
\st for every $i \geq 0$:
\begin{itemize}
\item $v_i + \delta \models \inv(\ell_i)$ for $0 \leq \delta \leq
  \delta_i$ (Def.~\ref{def-ta} implies that $v_i + \delta_i \models
  \inv(\ell_i)$ is equivalent);
\item there is a transition $(\ell_i,g_i,a_{i+1},r_i,\ell_{i+1}) \in
  E$ \st: ($i$) $v_i + \delta_i \models g_i$ and ($ii$)
  $v_{i+1}=(v_i+\delta_i)[r_i]$ (by the previous condition we have
  $v_{i+1} \models \inv(\ell_{i+1})$.)
\end{itemize}
If $\varrho$ is finite and ends in $s_n$, we let $\tgt(\varrho)=s_n$.
We say that event $a \in \Sigma \cup \{ \varepsilon\}$ is
\emph{enabled} in $s=(\ell,v)$, written $a \in \enabled(s)$, if there
is a transition $(\ell,g,a,R,\ell') \in E$ s.t.  $v \models g$ and
$v[R] \models \inv(\ell')$.
The set of finite (resp. infinite) runs 
from a state $s$ is denoted $\runs^*(s,A)$ (resp. $\runs^\omega(s,A)$)
and we define $\runs^*(A)=\runs^*((l_0,\vect{0}),A)$ and
$\runs^\omega(A)=\runs^\omega((l_0,\vect{0}),A)$.

We make the following \emph{boundedness} assumption on timed automata:
time-progress in every location is bounded. This is not a restrictive
assumption as every timed automaton that does not satisfy this
requirement can be transformed into a language-equivalent one that is
bounded~\cite{behrmann-hscc-01}. This implies that every infinite run
has an infinite number of events. We further assume\footnote{Otherwise
  the trace of an infinite word can have a finite number of events in
  $\Sigma$ but still infinite duration which cannot be defined in our
  setting. This is not a compulsory assumption and can be removed at
  the price of longer (not more complex) proofs.} that every
infinite run has an infinite number of discrete transitions with $a
\neq \varepsilon$.

The \emph{trace}, $\tr(\varrho)$, of a run $\varrho$ is the timed word
$\delta_0 a_1 \delta_1 a_2 \cdots a_n \delta_n \cdots$ where
$\varepsilon$ is removed (and durations are updated accordingly).  We
let $\dur(\varrho)=\dur(\tr(\varrho))$.  For $V \subseteq \runs^*(A)
\cup \runs^\omega(A)$, we let $\Tr(V)=\{\tr(\varrho) \ | \ \textit{
  $\varrho \in V$}\}$. 

A finite (resp. infinite) timed word $w$ is \emph{accepted} by $A$ if
$w = \tr(\varrho)$ for some $\varrho \in \runs^*(A)$ that ends in an
$F$-location (resp. for some $\varrho \in \runs^\omega(A)$ that
reaches infinitely often an $R$-location).  $\lang^*(A)$
(resp. $\lang^\omega(A)$) is the set of traces of finite
(resp. infinite) timed words accepted by $A$.
In the sequel we often omit the sets $R$ and $F$ in TA and this
implicitly means $F=L$ and $R=L$.

\subsection{Product of Timed Automata}
\begin{definition}[Product of TA] \label{def-prod-sync} Let
  $A_i=(L_i,l_0^i,X_i,\Sigma \cup \{ \varepsilon\},$ $E_i,\inv_i,F_i,R_i)$,
  $i \in\{1,2\}$, be TA \st $X_1 \cap X_2 = \emptyset$.  The
  \emph{product} of $A_1$ and $A_2$ is the TA $A_1 \times
  A_2=(L,l_0,X,$$\Sigma \cup \{ \varepsilon \},E,\inv,R,F)$ defined by:
$L=L_1 \times L_2$;
$l_0=(l_0^1,l_0^2)$;
$X = X_1 \cup X_2$; and
$E \subseteq L \times \calC(X) \times \Sigma \cup \{ \varepsilon \}\times 2^X \times
    L$ and
    $((\ell_1,\ell_2),g_{1,2},\sigma,r_{1,2},(\ell'_1,\ell'_2)) \in E$
    if:
    \begin{itemize}
    \item either $\sigma \neq \varepsilon$, and ($i$)
      $(\ell_k,g_k,\sigma,r_k,\ell'_k) \in E_k$ for $k=1$ and $k=2$;
      ($ii$) $g_{1,2} = g_1 \wedge g_2$ and ($iii$) $r_{1,2}=r_1 \cup r_2$;
    \item or $\sigma = \varepsilon$ and for $ k \in \{1,2\}$, ($i$)
      $(\ell_k,g_k,\sigma,r_k,\ell'_k) \in E_k$; ($ii$) $g_{1,2}=g_k$, ($iii$) $r_{1,2}=r_k$
      and ($iv$) $\ell'_{3-k}= \ell_{3-k}$;
    \end{itemize}
    $\inv(\ell_1,\ell_2)= \inv(\ell_1) \wedge \inv(\ell_2)$, $F = F_1
    \times F_2$ and $R$ is defined~\footnote{The product of B\"uchi
      automata requires an extra variable to keep track of the
      automaton that repeated its state. For the sake of simplicity we
      ignore this and assume the set $R$ can be defined to ensure
      $\lang^\omega(A_1) \cap \lang^\omega(A_2) = \lang^\omega(A_1
      \times A_2)$.} such that $\lang^\omega(A_1) \cap
    \lang^\omega(A_2) = \lang^\omega(A_1 \times A_2)$.
     \endef
\end{definition}

\subsection{Finite Automata}
A finite automaton (FA) is a TA with $X=\emptyset$: guards and
invariants are vacuously true and time elapsing transitions do not
exist.

We write $A=(L,$ $l_0,\Sigma \cup \{\varepsilon\},E,F,R)$ for a FA.  A
run of a FA $A$ is thus a sequence of the form: $\varrho = \ell_0
\xrightarrow{\ a_1\ } \ \ell_1 \cdots \ \ \ \cdots \xrightarrow{\
  a_{n}\ } \ \ell_n \cdots $
where for each $i \geq 0$, $(\ell_i,a_{i+1},\ell_{i+1}) \in E$.
Definitions of traces and languages are inherited from TA but the
duration of a run $\varrho$ is the number of steps (including
$\varepsilon$-steps) of $\varrho$: if $\varrho$ is finite and ends in
$\ell_n$, $\dur(\varrho)=n$ and otherwise $\dur(\varrho)=\infty$.  The
product definition also applies to finite automata.

%% file: predic-problems-nsc.tex
\section{Predictability Problems}
\label{sec-def-nsc}

Predictability problems are defined on partially observable TA.  Given
a TA $A=(L,\ell_0,$ $X,\Sigma,E,\inv,L,L)$, $\Sigma_o \subseteq
\Sigma$ a set of \emph{observable} events, and a \emph{bound} $\Delta
\in \setN$, we want to predict the occurrences of event $f \in \Sigma$
at least $\Delta$ time units before they occur. Without loss of
generality, we assume 1) that the target location of the
$f$-transitions is $l_f$, and they all reset a dedicated clock of $A$,
$x$, which is only used on $f$-transitions; 2) $A$ has transitions
$(l_f,\true,a,\{x\},l_f)$ for every $a \in \Sigma_o$.  We let
$\inv(l_f) = x \leq 1$.
In the remaining of this paper, $\Sigma_o$ is fixed and we use
$\projs$ for $\proj{\Sigma_o}$.

We again make the assumption that every infinite run of $A$ contains
infinitely many $\Sigma_o$ events: this is not compulsory but
simplifies some of the proofs.

\subsection{$\Delta$-Predictability}

A run $\rho$ of $A$ is \emph{non-faulty} if $\untimed(\tr(\rho))$ does
not contain event $f$; otherwise it is \emph{faulty}.  We write
$\nonfaulty(s,A)$ for the non-faulty runs from $s$
and define $\nonfaulty(A)= \nonfaulty((l_0,\vec{0}),A)$.
Let $\varrho \in \nonfaulty(A)$ be a finite non-faulty run:
\begin{eqnarray*}
  \varrho & = & (l_0,v_0) \xrightarrow{\delta_0} (l_0,v_0 + \delta_0)
  \xrightarrow{a_1} (l_1,v_1)  \cdots \xrightarrow{a_{n}} (l_n,v_n)
  \xrightarrow{\delta_n} (l_n,v_n+ \delta_n) \mathpunct.
\end{eqnarray*}
$\varrho$ is \emph{$\Delta$-prefaulty}, if it can be extended 
by a run $\varrho'$ as follows:
\[
\varrho'' = (l_0,v_0) \xrightarrow{\delta_0} \cdots
\xrightarrow{\delta_n} \underbrace{\tgt(\varrho) \xrightarrow{\
    \delta'_0\ } s'_1 \xrightarrow{\ a'_1 \delta'_1\ } \cdots
  \xrightarrow{\ a'_k \delta'_k\ } \cdots \xrightarrow{\ a'_j
    \delta'_j\ } s_j}_{\text{\normalfont run $\varrho'$}}
\]
where the extended run $\varrho'' \in \nonfaulty(A)$ satisfies: ($i$)
$f \in \enabled(s_j)$ and ($ii$) $\dur(\rho') \leq \Delta$ (\ie
$\sum_{k=0}^{j} \delta'_k \leq \Delta$.)  In words, $f$ can occur
within $\Delta$ time units from $\tgt(\varrho)$.
We let $\prefaulty_{\leq \Delta}(A)$ be the set of $\Delta$-prefaulty
runs of $A$.  Note that if $\Delta \leq \Delta'$ then
$\prefaulty_{\leq 0}(A) \subseteq \prefaulty_{\leq \Delta}(A)
\subseteq \prefaulty_{\leq \Delta'}(A)$.

We want to predict the occurrence of event $f$ at least $\Delta$ time
units before it occurs and it makes sense only if $\Delta \leq
\kappa(A)$ where $\kappa(A)$ is the minimum duration to reach a state
where $f$ is enabled.  If $f$ is never enabled, we let $\kappa(A) =
\infty$.
If $\kappa(A)$ is finite, let $0 \leq \Delta \leq \kappa(A)$ and
define the following timed languages: 
\begin{eqnarray}
  \label{eq-def-timed-lang-1}
  L^\omega_{\neg f} & = & \lang^\omega(A) \cap \Tr(\nonfaulty(A))  \\
  L^{-\Delta}_{f} & = & \Tr(\prefaulty_{\leq \Delta}(A))
   \label{eq-def-timed-lang-2}
  \mathpunct.
\end{eqnarray}
If $\kappa(A) = \infty$ then we let $L^{-\Delta}_{f} = \varnothing$.
$L^\omega_{\neg f}$ contains the infinite non-faulty traces of $A$.
$L^{-\Delta}_{f}$ contains the finite traces $w$ of $A$ that can be
extended into $w.x.f$ with $f$ occurring less then $\Delta$ time units
after $w$.

\medskip A \emph{$\Delta$-Predictor} is a device that predicts the
occurrence of $f$ at least $\Delta$ time units before it occurs.  It
should do it observing only the projection $\projs(w)$ of the current
trace $w$. Thus for every word $w \in L^{-\Delta}_f$, the predictor
predicts $f$ by issuing a $1$.  On the other hand, if a trace $w$ can
be extended as an infinite trace without any event $f$, \ie it is in $
\prefix{L^\omega_{\neg f}}$, the predictor must not predict $f$ and
thus should issue a $0$.  For a trace which is in $L^{-\Delta'}_f$
with $\Delta ' > \Delta$ and not in $ \prefix{L^\omega_{\neg f}}$, we do not
require anything from the predictor: it can predict $f$ or not
and this is why we define a predictor as a partial mapping.
\begin{definition}[$\Delta$-Predictor]\label{def-predictability}
  A \emph{$\Delta$-predictor} for $A$ is a partial mapping $P:
  \tw^*(\Sigma_o) \longrightarrow \{0,1\}$ such that:
  \begin{itemize}
  \item $\forall w \in L^{-\Delta}_f, P(\projs(w)) = 1$,
  \item $\forall w \in \prefix{L^\omega_{\neg f}}, P(\projs(w))=0$.
  \end{itemize}
  $A$ is \emph{$\Delta$-predictable} if there exists a $\Delta$-predictor for
  $A$ and is \emph{predictable} if there is some $\Delta$ such that $A$ is
  $\Delta$-predictable.
\endef
\end{definition}
It follows that if $f$ is never enabled in $A$, $A$ is
$\Delta$-predictable for any $\Delta$: a predictor is a mapping
$P(\cdot)=0$.
In the sequel we assume that $A$ contains a state where $f$ is enabled
and thus $\kappa(A)$ is finite.\footnote{Checking whether a state
  where $f$ is enabled is reachable and the computation of $\kappa(A)$
  can be done in PSPACE~\cite{courcoubetis-fmsd-92} for TA and linear
  time for FA.}

  In the dual problem of \emph{diagnosability}~\cite{tripakis-02}, it
  is required that the infinite words in $L^\omega_{\neg f}$ be
  \emph{non-Zeno}.  This is required by the problem statement that
  time must advance beyond any bound.  For predictability, this is not
  a requirement and we could accept non time-divergent runs in
  $L^{\omega}_{\neg f}$.  However for realistic systems we should add
  this requirement. This can be easily done and we discuss how to do
  this in section~\ref{sec-zeno}.
%
\subsection{PSPACE-Hardness of Bounded Predictability}
We are interested in the two following problems:
\begin{prob}[$\Delta$-Predictability (Bounded Predictability)] \label{prob-delta-pred} \mbox{} \\
  \textsc{Input:} A  TA $A=(L,\ell_0,X,\Sigma,E,\inv)$ and $\Delta \in \setN$. \\
  \textsc{Problem:} Is $A$ $\Delta$-predictable?
\end{prob}
\begin{prob}[Predictability] \label{prob-pred} \mbox{} \\
  \textsc{Input:} A TA $A=(L,\ell_0,X,\Sigma,E,\inv)$. \\
  \textsc{Problem:} Is $A$ predictable?
\end{prob}
Notice that predictability problems for finite automata are defined
using the number of steps in the automaton $A$ (including unobservable
steps) for the duration of a run.  
A first result is the PSPACE-hardness of the Bounded
Predictability problem. This is obtained by reducing the
\emph{reachability problem} for TA to the Bounded Predictability
problem.
The \emph{location reachability problem} for TA asks, given a location
$l$, whether $(l,v)$ (for some valuation $v$) is reachable from the
initial state of $A$.  This problem is PSPACE-complete for
TA~\cite{AlurDill94}.  
\begin{theorem}\label{rem-reduce-reach}
  The Bounded Predictability problem is PSPACE-hard for TA.
\end{theorem}
\begin{proof}
We can reduce the location reachability problem
for bounded TA to the predictability problem as follows (the reduction
is similar to~\cite{tripakis-02}): let $A$ be a bounded TA and $l$ a
location of $A$.  We can build $A'$ by adding transitions to $A$: let
END by a new location.  We add a transition
$(l,\true,f,\{x\},\text{END})$, and another one
$(l,\true,u,\{x\},\text{END}$) with $u$ unobservable, assuming $A$ has
at least one clock $x$.  We then add loops on location END
$(\text{END},x = 1,a,\{x\},\text{END})$, for each $a \in \Sigma$.
Moreover $\inv(\text{END})=x \leq 1$.  It follows from our definition
of predictability that $l$ is reachable in $A$ iff $A'$ is not
predictable, and $A'$ has size polynomial in $A$.
%
\end{proof}

\subsection{Necessary and Sufficient Condition for
  $\Delta$-Predictability}
We now give a necessary and sufficient condition (NSC) for
$\Delta$-predictability which is similar in form to the condition used
for \emph{$\Delta$-diagnosability}~\cite{tripakis-02}.

\begin{lemma}\label{lemma-nsc-fin}
  $A$ is $\Delta$-predictable iff $\projs(L^{-\Delta}_f) \cap
  \projs(\prefix{L^\omega_{\neg f}}) = \varnothing \mathpunct.$
\end{lemma}

\begin{proof}
\noindent{\it Only If.}
Assume $A$ is $\Delta$-predictable. There exists a partial mapping $P$
s.t.  $\forall w \in L^{-\Delta}_f, P(\projs(w)) = 1$, $\forall w \in
\prefix{L^\omega_{\neg f}}, P(\projs(w))=0$.  Assume $w \in
\projs(L^{-\Delta}_f) \cap \projs(\prefix{L^\omega_{\neg f}}) \neq
\varnothing$. Then $w = \projs(w_1) = \projs(w_2)$ with
$w_1 \in L^{-\Delta}_f$ and $w_2 \in \prefix{L^\omega_{\neg f}}$.
By definition of $P$ we must have $P(w) = P(\projs(w_1))=1$ and
$P(w) = P(\projs(w_2))=0$ which is a contradiction.

\noindent{\it If.}
If $\projs(L^{-\Delta}_f) \cap \projs(\prefix{L^\omega_{\neg f}}) =
\varnothing $ 
define $P(w) =1$ if $w \in \projs(L^{-\Delta}_f)$ and $P(w)=0$
otherwise.  If $P$ does not exist, we must have $w = \projs(w_1) =
\projs(w_2)$ with $w_1 \in L^{-\Delta}_f$ and $w_2 \in
\prefix{L^\omega_{\neg f}}$. In this case $w \in \projs(L^{-\Delta}_f)
\cap \projs(\prefix{L^\omega_{\neg f}})$ which is a contradiction.
\qed
\end{proof}
From Lemma~\ref{lemma-nsc-fin} we can prove the following Proposition
and Theorem:
\begin{proposition}\label{prop-lt}
  if $\Delta \leq \Delta'$ and $A$ is $\Delta'$-predictable, then
  $A$ is $\Delta$-predictable.
\end{proposition}
\begin{proof}
  $L^{-\Delta}_f \subseteq L^{-\Delta'}_f$ and thus
  $\projs(L^{-\Delta}_f) \cap \projs(\prefix{L^\omega_{\neg f}})
  \subseteq \projs(L^{-\Delta'}_f) \cap \projs(\prefix{L^\omega_{\neg
      f}})$.  \qed
\end{proof}

\begin{theorem}\label{thm-pred}
  $A$ is predictable iff $A$ is $0$-predictable.
\end{theorem}
In the next section, we focus on the $\Delta$-predictability problem
for finite automata and discuss how it generalizes the previous notion
introduced by S.~Genc and S.~Lafortune in~\cite{automatica-genc-09}.
Section~\ref{sec-predic-ta} tackles the $\Delta$-predictability
problem for TA.


%% file: predic-fa.tex
\section{Predictability for Discrete Event Systems}
\label{sec-predic-fa}
In this section, we address the predictability problems for discrete
event systems specified by FA. We first show that the definition of
predictability (Def.~\ref{def-predictability}) we introduced in
Section~\ref{sec-def-nsc} is equivalent to the original definition of
predictability by S.~Genc and S.~Lafortune
in~\cite{automatica-genc-09}.
 
\subsection{Original Definition of Predictability (S.~Genc and S.~Lafortune)}
Let $L_f = \Tr(\prefaulty_{\leq 0}(A))$ be the set of non-faulty
traces that can be extended with a fault in one step, and $L_{\neg f}
= \prefix{\Tr(\nonfaulty(A))}$ be the set of finite prefixes of
non-faulty traces.
S.~Genc and S.~Lafortune originally defined predictability for
discrete event systems in~\cite{automatica-genc-09} and we refer to
GL-predictability for this definition.  GL-predictability is defined
as follows\footnote{Technically S.~Genc and S.~Lafortune let $w$ range
  over $L_f.f$ and impose that $|t|_f = 0$; the definition we give in
  Equation~\eqref{eq-pred-lafortune} is equivalent to Definition~1
  of~\cite{automatica-genc-09}.}:
\begin{eqnarray} \label{eq-pred-lafortune} \exists n \in \setN,
  \forall w \in L_f, \exists t \in \overline{w} \text{ such that } {\mathbf P}(t)
\end{eqnarray}
with ${\mathbf P}(t)$ defined by:
\begin{eqnarray*} \label{eq-p} 
  \mathbf{P}(t) : \forall u \in L_{\neg f}, \forall v \in \mathcal{L}(A)/u,
  \projs(u)=\projs(t) \wedge |v| \geq n \implies |v|_f > 0 \mathpunct.
\end{eqnarray*}
According to~\cite{automatica-genc-09}, $A$ is GL-predictable iff
Equation~\eqref{eq-pred-lafortune} is satisfied.
GL-predictability as defined by
Equation~\eqref{eq-pred-lafortune} is equivalent to our notion of
predictability:
\begin{theorem}\label{thm-equiv}
  $A$ is $G
L$-predictable iff $A$ is $0$-predictable.
\end{theorem}

\subsection{Checking $k$-Predictability}
To check whether $A$ is $k$-predictable, $0 \leq k \leq \kappa(A)$, we
can use the NSC we established in Lemma~\ref{lemma-nsc-fin}: $A$ is
$k$-predictable iff $\projs(L^{-k}_f) \cap
\projs(\prefix{L^\omega_{\neg f}}) = \varnothing \mathpunct.$ To check
this condition, it suffices to build a \emph{twin plant} (similar
to~\cite{automatica-genc-09} and to what is defined for fault
diagnosis~\cite{yoo-lafortune-tac-02}).
We define two automata $A_1(k)$ and $A_2$ that accept
$\projs(L^{-k}_f)$ and $\projs(\prefix{L^\omega_{\neg f}})$ and
synchronize them to check whether the intersection is empty.
%
%
The first automaton $A_1(k)$ accepts finite words 
which are in $\projs(L_f^{-k})$ and is defined as follows:
\begin{enumerate}
\item in $A$, we compute the set of states $F_k$ that can reach a
  state where $f$ is enabled within $k$ steps (this can be done in
  linear time using a backward breadth-first search from states where
  $f$ is enabled.)
\item $A_1(k)$ is a copy of $A$ where the set of final states is $F_k$,
  and every $a \not\in \Sigma_o$ is replaced by $\varepsilon$.
\end{enumerate}
It follows that $A_1(k)$ accepts $\projs(L^{-k}_f)$.
 

The second automaton $A_2$ accepts $\projs(\prefix{L^\omega_{\neg
    f}})$.  To compute it, we merely need to compute the states from
which there is an infinite path without any state where $f$ is
enabled.  This can be done in linear time again (\eg computing the
states that satisfy the CTL formula $\textsf{EG} \neg \enabled(f)$.)
$A_2$ is defined as follows:
\begin{enumerate}
\item let $F_{\neg f}$ be the set of states in $A$ from which there
  exists an infinite path with no states where $f$ is enabled.
\item $A_2$ is a copy of $A$ restricted to the set of states $F_{\neg
    f}$, and every $a \not\in \Sigma_o$ is replaced by $\varepsilon$
  (this implies that the target state of the $f$ transitions cannot be
  in $A_2$).
\end{enumerate}
From the previous construction with sets of accepting states $F_k$ for
$A_1(k)$ and $F_{\neg f}$ for $A_2$ (every state in $A_2$ is
accepting), $\lang^*(A_1(k) \times A_2) = \projs(L^{-k}_f) \cap
\projs(\prefix{L^\omega_{\neg f}})$ and we can check
$k$-predictability in quadratic time in the size of $A$.

\begin{example}
  For the untimed version of Automaton $G$ (Fig.~\ref{fig-example},
  page~\pageref{fig-example}), we obtain $G_1(0)$ and $G_2$ as
  depicted on Fig.~\ref{fig-construction}.  Recall that $d$ is
  unobservable.
\end{example}
\begin{figure}[hbtp]
  \centering
\subfigure[$G_1(0)$]{
 \begin{tikzpicture}[thick,node distance=0.8cm and .8cm]%
    \small
    \node[state,initial] (q_0)  {$l_0$}; 
    \node[state] (q_1) [below right =of q_0] {$l_1$};
    \node[state,acc] (q_2) [below right=of q_1] {$l_2$};  
    \path[->] (q_0) edge node {$a$} (q_1)
    (q_1) edge node {$c$} (q_2)
    (q_2) edge[loop left] node {$a,b,c$} (q_2); 
  \end{tikzpicture}
}
\subfigure[$G_2$]{
\begin{tikzpicture}[thick,node distance=0.8cm and .8cm]%
    \small
    \node[state,initial] (q_0)  {$l_0$}; 
    \node[state] (q_3) [below right=of q_0] {$l_3$};
    \node[state,acc] (q_4) [below right=of q_3] {$l_4$};  
    \path[->] (q_0) edge node {$\varepsilon$} (q_3)
    (q_3) edge node {$a$} (q_4)
    (q_4) edge[loop left] node {$b$} (q_4); 
  \end{tikzpicture}
}
\subfigure[$G_1(0) \times G_2$]{
\begin{tikzpicture}[thick,node distance=0.8cm and 1.2cm]%
    \small
    \node[state,initial,rectangle, rounded corners] (q_0)  {$l_0,l_0$}; 
    \node[state,rectangle, rounded corners] (q_3) [below right=of q_0] {$l_0,l_3$};
    \node[state,rectangle, rounded corners] (q_4) [below right=of q_3] {$l_1,l_4$};  
    \path[->] (q_0) edge[swap] node {$\varepsilon$} (q_3)
    (q_3) edge[swap] node {$a$} (q_4);
  \end{tikzpicture}
}
  \caption{Construction of $G_1$ and $G_2$ for automaton $G$
    ((Fig.~\ref{fig-example})}
  \label{fig-construction}
\end{figure}

Computing the largest $k$ such that $A$ is $k$-predictable can also be
done in quadratic time.  In $A$, we can compute, in linear
time\footnote{\eg standard breadth-first search~\cite{cormen-2009} on
  $A$.}, the shortest distance $d_f(q)$ (going backwards) from $q$ to
a state where $f$ is enabled (it is $\infty$ if $q$ is unreachable
going backwards in $A$).  In the product $A_1(k) \times A_2$, if there
is a run from the initial state to $(s_1,s_2)$ and $d(s_1) = k', k'
\leq k$, this implies that $A$ is not $k'$-predictable.  To determine
the largest $k$ such that $A$ is $k$-predictable, it suffices to
perform the following steps:
\begin{enumerate}
\item compute the shortest distance $d_f(q)$ to an $f$-enabled state
  for each $q \in Q$;
\item build the product $A_1(0) \times A_2$;
\item let $S$ be the set of reachable states in $A_1(0) \times A_2$
  and $M = \min_{(s_1,s_2) \in S} d_f(s_1)$.
\end{enumerate}
The largest $k$ such that $A$ is $k$-predictable is $M-1$.


\begin{example}
  On automaton $G$ of Fig.~\ref{fig-example}: $d(l_2)=0$, $d(l_1)=1$,
  $d(l_0)=2$, $d(l_3)=d(l_4)=\infty$.  The minimum value reachable in
  $G_1(0) \times G_2$ is obtained for $l_1$ and is $d(l_1)=1$. Thus
  $G$ is $0$-predictable.
\end{example}


%% file: predic-ta.tex
\section{Predictability for Timed Automata}
\label{sec-predic-ta}

In this section we address the predictability problems for TA.  We
first rewrite the NSC of Lemma~\ref{lemma-nsc-fin} using infinite
languages.  This enables us 1) to deal with time-divergent runs and 2)
to design an algorithm to solve the predictability problems for TA.

\subsection{Checking $\Delta$-Predictability}
\label{sec-check-pred-ta}
We can reformulate Lemma~\ref{lemma-nsc-fin} without the prefix
operator by extending $L^{-\Delta}_f$ into an equivalent language of
infinite words: let $L^{\omega,-\Delta}_f = L^{-\Delta}_f.(\Sigma_o
\times \setR_{\geq 0} )^\omega$.
\begin{lemma}\label{lemma-inf}
 $ \projs(L^{-\Delta}_f) \cap
  \projs(\prefix{L^\omega_{\neg f}}) = \varnothing
  \iff \projs(L^{\omega,-\Delta}_f) \cap
  \projs(L^\omega_{\neg f}) = \varnothing$.
\end{lemma}

\medskip

To check $\Delta$-predictability we build a product of timed automata
$A_1(\Delta) \times A_2$, and reduce the problem to B\"uchi emptiness
on this product.  This construction is along the lines of the
\emph{twin plant} introduced
in~\cite{yoo-lafortune-tac-02,tripakis-02}.  The difference in the
predictability problem lies in the construction of $A_1(\Delta)$ which is
detailed later.  The twin plant idea is the following:
\begin{itemize}
\item $A_1(\Delta)$ accepts $\projs(L^{\omega,-\Delta}_f)$ \ie
  (projections of) infinite timed words of the form $w.(\setR_{\geq 0}
  \times \Sigma_o)^\omega$ with $w \in L^{-\Delta}_f$;
\item $A_2$ accepts $\projs(L^\omega_{\neg
    f})$ 
  \ie (projections of) infinite non-faulty timed words in
  $L^\omega_{\neg f}$;
\item the product $A_1(\Delta) \times A_2$ accepts the language
  $\projs(L^{\omega,-\Delta}_f) \cap \projs(L^\omega_{\neg f})$;
\item thus checking $\Delta$-predictability of $A$ reduces to B\"uchi
  emptiness checking on the product $A_1(\Delta) \times A_2$.
\end{itemize}
%

\begin{figure}[thbtp] 
  \centering
  \scalebox{0.7}{
  \begin{tikzpicture}[thick,node distance=2.3cm and 2.6cm]%
\tikzset{initial text={},
    every state/.style={rectangle, rounded corners,
      minimum size=.6cm,draw=black!50,very thick,
    },
    acc/.style={double distance=1.5pt},
    on grid,auto,
    bend angle=25}
    \small

    \node[state] (l) {$l$};
    \node[state,right=of l] (l1) {$l_1$};
    \node[state,below=of l] (ltilde) {$\tilde{l}$};
    \node[state,right=of ltilde] (l1tilde) {$\tilde{l}_1$};

    \node[state,right=of l1,xshift=1cm] (lprime) {$l'$};
    \node[state,right=of lprime] (lf) {$l_f$};

   \node[state,right=of l1tilde,xshift=1cm] (lprimetilde) {$\tilde{l}'$};
    \node[state,right=of lprimetilde,label=below:{$[y \leq 1]$}] (end) {END};
    \node[state,right=of end,label=below:{$[y \leq 0]$}] (nz) {NZ};

    \node[left=of l,xshift=1cm,yshift=.3cm] (foo1) {};
    \node[below=of foo1] (foo2) {};

    \node[right=of l1,xshift=-.8cm,yshift=.3cm] (foo3) {}; 
    \node[below=of foo3] (foo4) {}; 
    
    \node[right=of lprime,xshift=-2.2cm] (foo5) {\LARGE $\sf X$};

    \path[->] (foo1) edge[pos=0.1] node {$g_1,a_1,r_1$} (l)   
    (foo2) edge[pos=0.1] node {$g_1,\varepsilon,r_1$} (ltilde)
    (l) edge node {$g_2,a_2,r_2$} (l1)
    (ltilde) edge node {$g_2,\varepsilon,r_2$} (l1tilde)
    (lprime) edge node {$g,f$} (lf)
    (lprimetilde) edge node {\stack{$g \wedge y \leq \Delta,\varepsilon$\\$y:=0$}} (end)
    (end) edge[bend left] node {\stack{$y=1,\varepsilon$\\$y:=0$}} (nz)
    (nz) edge[bend left] node {$\varepsilon$} (end)
    (l1) edge[dashed] (foo3)
    (foo3) edge[dashed] (lprime)
    (l1tilde) edge[dashed] (foo4)
    (foo4) edge[dashed] (lprimetilde)
    (ltilde) edge[loop below] node[label=right:{$\Sigma_o$}]  {} (ltilde)
    (l1tilde) edge[loop below] node[label=right:{$\Sigma_o$}] {}  (1ltilde)
    (lprimetilde) edge[loop below] node[label=right:{$\Sigma_o$}] {} (lprimetilde)
    (end) edge[loop above] node {$\Sigma_o$} (end)
    (nz) edge[loop above] node {$\Sigma_o$} (nz)
    (l) edge[bend right,swap,pos=0.4] node {$y:=0$} (ltilde)
    (l1) edge[bend right,swap,pos=0.4] node {$y:=0$} (l1tilde)
    (lprime) edge[bend right,swap,pos=0.4] node {$y:=0$} (lprimetilde)
    ;

\coordinate (nz') at ($(nz.west) + (0.4,1.35)$);
\coordinate (foo2') at ($(ltilde.east) - (0.1,0.9)$);
\coordinate (nz') at ($(nz.west) + (0.4,1.2)$);
\begin{pgfonlayer}{background}
\node [draw=blue!80!black,fill=blue!10!white,opacity=0.8,
rounded corners, xshift=.2cm, yshift=0.05cm, fit= (foo2) (foo2') (nz) (nz')] {};
\end{pgfonlayer}

  \end{tikzpicture}
}
\caption{Construction of Automaton $A_1$}
\label{fig-a1}
\end{figure}

$A_1(\Delta)$ itself is made of two copies of $A$: the original $A$
and a twin copy (see Fig.~\ref{fig-a1}). $A_1$ starts in the initial
location of $A$, $\ell_0$, and at some point in time switches to the
twin copy (grey area on Fig.~\ref{fig-a1}).  The purpose of the twin
copy is to extend the previously formed timed word with a timed word
of duration less than $\Delta$ time units that reaches a state where
$f$ is enabled. The actions performed in the copy do not matter as we
only have to check that $f$ is reachable within $\Delta$ time units
since we switched to the copy.  In this case the timed word built in
the original $A$ is in $L^{-\Delta}_f$.

$A_1(\Delta)=(L \cup \tilde{L} \cup \{ \text{END} \}, l^1_0,X \cup \{
y \}, \Sigma \cup \{ \varepsilon\}, E_1, \inv_1, \varnothing, \{
\text{END} \})$ is formally defined as follows\footnote{For now ignore
  the NZ location in the Figure and the invariants $[y \leq k]$. Their
  sole purposes is to ensure time-divergence.} (see
Fig.~\ref{fig-a1}):
\begin{itemize}
\item $\tilde{L} = \{ \tilde{\ell}, \ell \in L\}$ is the set of twin
  locations; 
\item $l^1_0 = l_0$; $A_1(\Delta)$ starts in the same initial state as
  $A$.
\item $ \inv_1(\ell) = \inv_1(\tilde{\ell}) = \inv(\ell)$; invariants are the
  same as in the original automaton $A$ including the twin locations;
\item the transition relation is defined as follows:
  \begin{itemize}
  \item original transitions of $A$: $(\ell,g,a',R,\ell') \in E_1$ iff
    $(\ell,g,a,R,\ell') \in E$ and $a \in \Sigma_o \setminus \{ f \}$;
    $a'=a$ if $a \in \Sigma_o$ and $a'=\varepsilon$ otherwise; this
    renaming hides the unobservable events by renaming them in
    $\varepsilon$.
  \item transitions to the twin locations:
    $(\ell,\true,\varepsilon,\{y\},\tilde{\ell}) \in E_1$ for each $\ell
    \in L$; $A_1$ can switch to the twin copy at any time and doing so
    preserves the values for the clocks in $X$ but resets $y$;
  \item equivalent unobservable transitions inside the twin copy:
    $(\tilde{\ell},g,\varepsilon,R,\tilde{\ell}_1) \in E_1$ iff
    $(\ell,g,a,R,\ell_1) \in E$ for some $a \neq f$;
  \item equivalent of $f$-transitions in the twin copy:
    $(\tilde{\ell}',g \wedge y \leq \Delta,\varepsilon,R,\text{END})
    \in E_1$ iff $(\ell',g,f,R,l_f) \in E$.
  \item loop transitions on observable events in the twin copy:
    $(\tilde{\ell},\true,a,\varnothing,\tilde{\ell}) \in E_1$ for each
    $a \in \Sigma_o$. This enables $A_2$ (defined below) to synchronize
    with $A_1$ on $\Sigma_o$ after $A_1$ has chosen to switch to the
    twin copy of $A$.
  \end{itemize}
\end{itemize}
Finally, $A_2$ is simply of copy of $A$ without the $f$-transitions
and the clocks are renamed to be local to $A_2$. Every location in
$A_2$ is a \emph{repeated} location.  Notice that the only repeated
location in $A_1(\Delta)$ is END.  By definition of the synchronized
product, $\lang^\omega(A_1(\Delta) \times A_2) =
\lang^\omega(A_1(\Delta)) \cap \lang^\omega(A_2)$.
\begin{lemma}\label{lemma-equiv}
  $\projs(L^{\omega,-\Delta}_f) \cap \projs(L^\omega_{\neg f}) =
  \lang^\omega(A_1(\Delta) \times A_2)$.
\end{lemma}

\begin{theorem}
  Problems~\ref{prob-delta-pred} and~\ref{prob-pred} are
  PSPACE-complete.
\end{theorem}
\begin{proof}
  PSPACE-easiness of Problem~\ref{prob-delta-pred} is established as
  follows: checking B\"uchi emptiness for timed automata is in
  PSPACE~\cite{AlurDill94}.  The product $A_1(\Delta) \times A_2$ has
  size polynomial in the size of $A$ and thus checking B\"uchi
  emptiness of the product is in PSPACE as well.
  Problem~\ref{prob-delta-pred} is thus in PSPACE.
  By Theorem~\ref{thm-pred}, Problem~\ref{prob-pred} is in PSPACE
  as well.

  Theorem~\ref{rem-reduce-reach} states PSPACE-hardness for
  Problem~\ref{prob-delta-pred}.
%
  As $0$-predictability \ie Problem~\ref{prob-delta-pred}, is
  equivalent to Problem~\ref{prob-pred}, it is PSPACE-hard as well.
\qed
\end{proof}

\subsection{Restriction to Time-Divergent Runs of $L^{\omega}_{\neg
    f}$}\label{sec-zeno}
To deal with time-divergence and enforce the runs in $L^{\omega}_{\neg
  f}$ to have infinite duration (see Remark~\ref{rem-zeno}), we can
add another automaton in the product with a B\"uchi condition that
enforces time-divergence (this is how this kind of requirements is
usually addressed).  In our setting, we can re-use the fresh clock $y$
of $A_1(\Delta)$ after location END is visited: it is not useful
anymore to check whether a timed word is in $L^{-\Delta}_f$.  The
modifications to $A_1(\Delta)$ required to ensure time-divergence in
$A_2$ are the following:
\begin{itemize}
\item add a new location NZ, which is now the repeated location of
  $A_1(\Delta)$;
\item add two transitions as depicted on Fig.~\ref{fig-a1} between END
  and NZ.
\end{itemize}
This way infinite timed words accepted by $A_1(\Delta)$ must be
time-divergent and with the synchronization with $A_2$ this forces the
runs of $A_2$ to be time-divergent.

\medskip

Finally, once we know how to solve Problem~\ref{prob-delta-pred}, we
can compute the optimal (maximum) anticipation delay by performing a
binary search on the possible values of $0 \leq \Delta \leq
\kappa(A)$.


\subsection{Implementability of the $\Delta$-Predictor}
In the previous sections, we defined a predictor as a mapping from
timed words to $\{0,1\}$.  To build an implementation of this mapping
(an actual predictor) we still have some key problems to address: 1)
we have to recognize when a timed word is in $L^{-\Delta}_f$; and 2)
we have to detect that a timed word is in $L^{-\Delta}_f$ as soon as
possible.
S.~Tripakis addressed similar problems in~\cite{tripakis-02} in the
context of fault diagnosis where a \emph{diagnoser} is given as an
algorithm that computes a state estimate of the system after reading a
timed word $w$.  The diagnoser updates its status after the occurrence
of an observable event or after a \emph{timeout} (TO) has occurred,
which means some timed elapsed since the last update and no observable
event occurred.  The value of the timeout period (TO) is required to
be less than the minimum delay between two observable events to ensure
that the diagnoser works as expected.
However, point 2) above still poses problem in our context, as
demonstrated by the TA $\calB$ of Fig.~\ref{fig-ex-sample}.

\noindent The set of observable events is $\{a\}$ and $\calB$ is
$4$-predictable. To see this, define the predictor $P$ as follows: for
a timed word $w = \delta.w'$ with $\delta \geq 2$, $P(w)=1$ and
otherwise $P(w)=0$.  Indeed if $2$ time units elapse and we see no
observable events, for sure the system is still is $l_0$ and thus a
fault $f$ is bound to happen, but not before $4$ time units.  An
implementation of a $4$-predictor has to observe the state of the
system \emph{exactly} at time $2$ otherwise it cannot predict the
fault $4$ time units in advance.

\begin{floatingfigure}{5.2cm}
  \centering
  \begin{tikzpicture}[thick,node distance=1cm and 3cm]%
    \small
    \node[state,initial] (q_0) [label=-90:{$[x \leq 8]$}] {$l_0$}; 
    \node[state] (q_1) [below right=of q_0,label=-90:{$[x \leq 1]$}] {$l_1$};
    \node[state,secret] (q_2) [above right=of q_0] {};
    \path[->] (q_0) edge node {$6 \leq x \leq 8$,$f$} (q_2) 
              (q_0)  edge node[swap,pos=0.7]  {\stack{$x < 1$ \\ $\varepsilon$\\$x:=0$}} (q_1)
              (q_1) edge[loop above] node {$x=1$; $a$ ;$x:=0$} (q_1)
              (q_2) edge[loop above] node {$a$} (q_2);
  \end{tikzpicture}
\caption{The Timed Automaton $\calB$}
\label{fig-ex-sample}
\end{floatingfigure}
\noindent  Now assume the platform on which we
implement the predictor can make an observation every $\frac{3}{5}$
time units.  The first observation of the predictor occurs at time
$\frac{3}{5}$; the third at $\frac{9}{5}$ and we cannot predict the
fault as we still don't know whether the system is in $l_0$ or has
made a silent move to $l_1$.  The next observation is at
$\frac{12}{5}$: if we have seen no $a$ so far, for sure the system is
in $l_0$ and we can predict the fault. However the fault may now occur
in $\frac{18}{5}$ time units \ie less than $4$ time units from the
current time.  Such a platform cannot implement the $4$-predictor.
The maximal anticipation delay we computed in the previous section is
thus an \emph{ideal} maximum that can be achieved by an \emph{ideal}
predictor that could monitor the system \emph{continuously}.
In a realistic system, there is a \emph{sampling rate}, or at least a
minimum amount of time between two observations~\cite{dewulf-hscc-04}.
In the sequel we address the \emph{sampling predictability problem}
that takes into account the speed of the platform.

\subsection{Sampling Predictability}
Let $\alpha \in \setQ$ and $L$ be a timed language.  We let $L \!\!
\mod \!\alpha$ be the set of timed words in $L$ with a duration
multiple of $\alpha$: $L \mod \alpha = \{ w \in L, \exists k \in
\setN, \dur(w) = k \cdot \alpha\}$.

Given a \emph{sampling rate} $\alpha \in \setQ$, the \emph{sampling}
predictability problem is defined by refining the definition of a
$\Delta$-predictor: an \emph{$(\alpha,\Delta)$-predictor} for $A$ is a
partial mapping $P: \tw^*(\Sigma_o) \!\! \mod \!\alpha \longrightarrow
\{0,1\}$ such that:
\begin{itemize}
\item $\forall w \in L^{-\Delta}_f \!\!\mod \!\alpha, P(\projs(w)) = 1$,
\item $\forall w \in \prefix{L^\omega_{\neg f}} \!\!\mod \!\alpha, P(\projs(w))=0$.
\end{itemize}
A timed automaton $A$ is $(\alpha,\Delta)$-predictable if there exists
a $(\alpha,\Delta)$-predictor for $A$ and is $\alpha$-predictable is
there is some $\Delta$ such that $A$ is $(\alpha,\Delta)$-predictable.
\begin{remark}
  The problem of deciding whether there exists a sampling rate
  $\alpha$ such that $A$ is $\alpha$-predictable is also interesting
  but very likely to be undecidable as the existence of a
  \emph{sampling rate} s.t. a location is reachable in a TA is
  undecidable~\cite{cassez-hscc-02}.
\end{remark}
The solution to the sampling predictability problem is a simple
adaptation of the solution we presented in
Section~\ref{sec-predic-ta}: in the construction of automaton
$A_1(\Delta)$ (Fig.~\ref{fig-a1}, page~\pageref{fig-a1}), it suffices
to restrict the transitions from the original $A$ to the twin copy
(those resetting $y$) to happen at time points multiple of
$\alpha$. This can be achieved by adding a \emph{sampler} timed
automaton, and a common fresh clock, $s$, that \emph{sampler} resets
every $\alpha$ time units. The transitions resetting $y$ in $A_1$ are
now guarded by $s = 0$.

We can now safely define an implementation for an
$(\alpha,\Delta)$-predictor along the lines of the \emph{diagnoser}
defined in~\cite{tripakis-02}.  The implementation performs an
observation every $\alpha$ time units.  It computes a state estimate
of the system.  If one of the states in the state estimate can reach a
state where $f$ is enabled within $\Delta$ time units, the predictor
predicts $f$ and issue a $1$. Otherwise it issues $0$.  Computing a
representation of the state estimate as a set of polyedra is a
standard operation and can be done given an observed timed word $w$,
and the timed automaton model $A$.  Checking that one of the states in
the estimate can reach an $f$-enabled state within $\Delta$ time
units can also be done using standard reachability algorithm.  It can
be performed on-line or off-line by computing a polyedral
representation of this set of states.

\subsection{A Simple Example}
The example of Fig.~\ref{fig-ex-sample} 
can be analyzed using \uppaal~\cite{uppaal-sttt}. \uppaal cannot check
for B\"uchi emptiness but in this example there is no \emph{Zeno}
non-faulty behaviours; thus we can restrict to a sufficiently large
horizon to check the condition of Lemma~\ref{lemma-inf}. 

The construction of the product $\calB_1(\Delta) \times \calB_2$
defined in Section~\ref{sec-check-pred-ta} for $\calB$ is depicted on
Fig.~\ref{fig-uppaal-ex}.  Assume the sampling rate is $\alpha =
\frac{q}{p}$. The rational rates must be encoded by scaling up the
constants in a network of TA as \uppaal only accepts integers to
compare clocks against.
We use the variables \texttt{qsRate} and \texttt{psRate} in the
\uppaal model for these two constants.  To obtain a network of TA with
integers, and \emph{sampling rate} $\alpha$, we multiply all the
constants by $p$ (this is standard in TA and scales up time such that
one time unit in the original automaton is $p$ time units in the
scaled up one).  We add one automaton \emph{sampler} that resets the
clock $s$ every $q$ time units. The transitions in $\calB_1(\Delta)$
that reset $y$ are now guarded by $s = 0$ which implies there can only
be taken at points in time which are multiples of $q$.  As mentioned
earlier we cannot check a B\"uchi condition with \uppaal and replace
it by a reachability condition on a sufficiently large horizon. Note
also that the $\Delta$ ($D$ in the \uppaal model) is multiplied by $q$
in the guard leading to END.  Synchronization is realized with a
broadcast channel for each observable event.

Given a value of $D$, the property we check is $P$: ``Can we reach END
in the product with global time larger than $M * p$''?  $M=10$ is
enough for our example.  If the answer is ``yes'' then the system is
not $(D \cdot\frac{p}{q})$-predictable, otherwise it is.

For a sampling rate $\alpha=\frac{3}{5}$, we get as expected that the
maximum $D$ for which $\calB$ is predictable is $6$.  Which means that
the actual maximal anticipation delay is $\Delta = 6 \cdot \frac{3}{5}
= \frac{18}{5}$ time units. And indeed, the first time we can check
that more than $2$ time units have elapsed is $\frac{12}{5}$ and thus
an interval of $\frac{18}{5}$ before $f$ can occur.
If we set $\alpha=1$ we get $D = \Delta=4$ meaning we can ideally predict
the fault $4$ time units in advance.

\begin{figure}[hbtp]
\includegraphics[width=.98\textwidth]{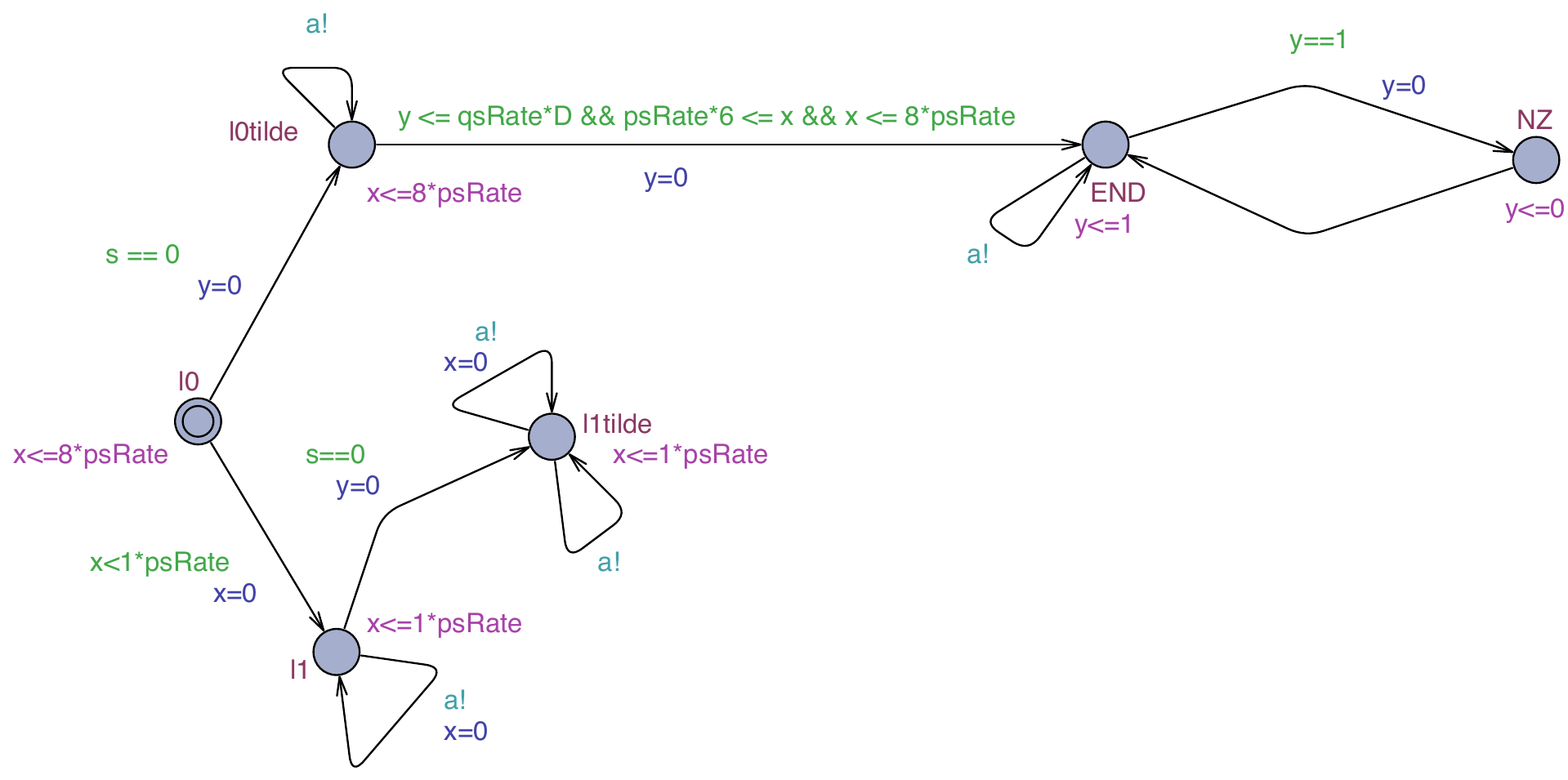}
\includegraphics[width=.6\textwidth]{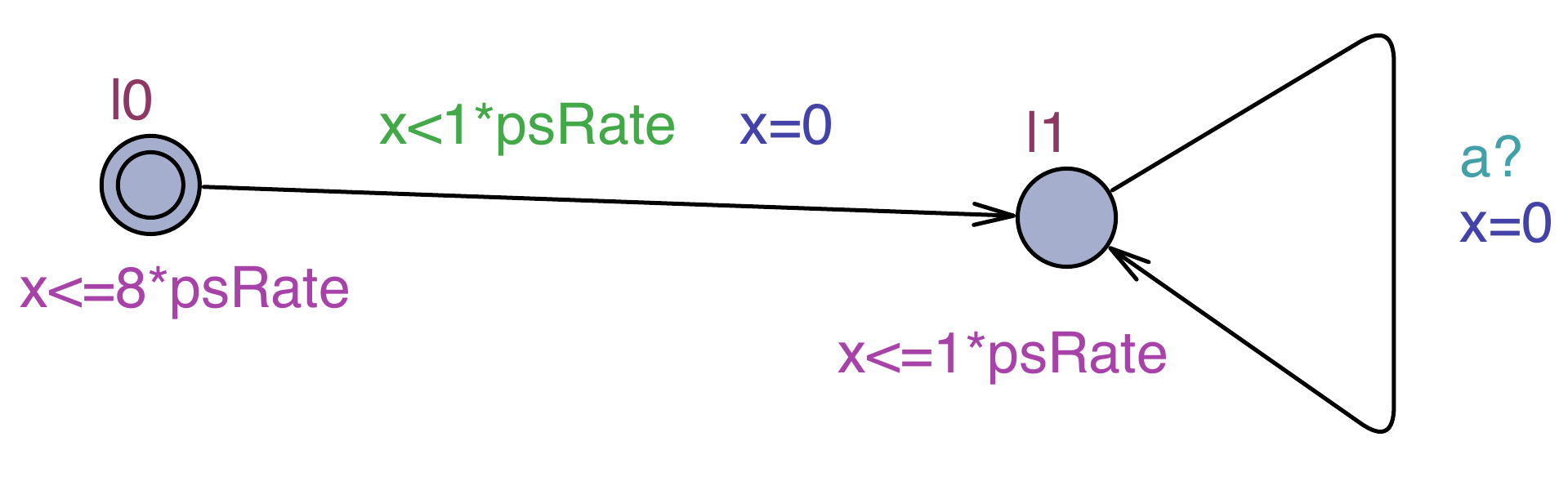} \hspace*{1cm}
\includegraphics[width=.3\textwidth]{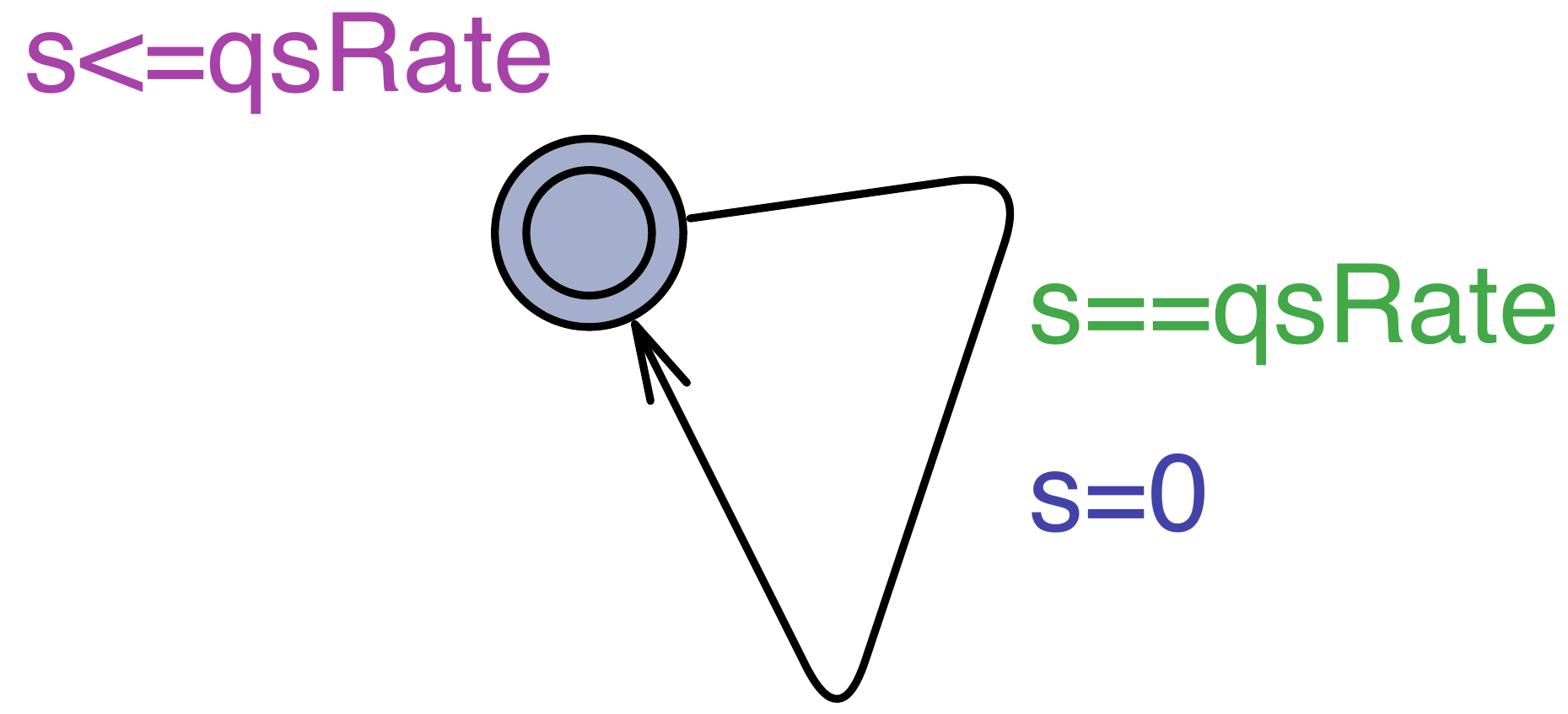}
\caption{\uppaal Models for $\calB$ of Fig.~\ref{fig-ex-sample}.}
\label{fig-uppaal-ex}
\end{figure}

%% file: conclusion.tex
\section{Conclusion and Future Work}
\label{sec-conclu}
In this paper we have proved some new results for predictability of
events' occurrences for timed automata.  We also contributed a new and
simpler definition of bounded predictability for finite automata.
The natural extensions of our work are as follows:
\begin{itemize}
\item in~\cite{Bouyerfossacs05}, P.~Bouyer, F.~Chevalier and
  D.~D'Souza proposed an algorithm to decide the existence of a
  diagnoser with fixed resources (number of clocks and constants).
  The very same question arises for the existence of a predictor in
  timed systems.
\item \emph{dynamic} observers~\cite{cassez-fi-08} have been proposed
  in the context of fault diagnosis and
  \emph{opacity}~\cite{cassez-fmsd-2012}; in~\cite{cassez-acsd-07} it
  is shown how to compute a most permissive observer that ensures
  diagnosability (or opacity~\cite{cassez-atva-2009}) and also how to
  compute an optimal observer~\cite{cassez-tase-07} (w.r.t. to a given
  criterion).  We can define the same problems for predictability.
\item given the similarities between the fault diagnosis and
  predictability problems, it would be interesting to state these two
  problems in a similar and unified way and design an algorithm that
  can solve the unified version.
\end{itemize}